\documentclass[english, a4paper]{article}

\usepackage[margin=1.65in]{geometry}
\usepackage[font=small, labelfont=bf, margin=0.35in]{caption}

\usepackage{titling}
\usepackage{authblk}
\usepackage{bm}

\usepackage{babel}
\usepackage[utf8]{inputenc}
\usepackage[pdftex]{hyperref}

\usepackage{xcolor}

\usepackage{amsmath}
\usepackage{amssymb}
\usepackage{amsthm}

\usepackage{cite}

\usepackage{multirow}

\usepackage[labelfont=bf, figurename=Fig., tablename=Tab.]{caption}

\allowdisplaybreaks

\newcommand{\refeq}[1]{(\ref{eq:#1})}

\newcommand{\ket}[1]{\left|\, #1 \, \right\rangle}

\newcommand{\ordo}{O} 
\newcommand{\ceil}[1]{\left\lceil #1 \right\rceil}

\newcommand{\norm}[1]{\left\lVert#1 \right\rVert}
\newcommand{\abs}[1]{\left\lvert#1 \right\rvert}
\newcommand{\vect}[1]{\bm{#1}}

\DeclareMathOperator{\mathpoly}{poly}
\newcommand{\poly}{{\mathpoly}}
\DeclareMathOperator{\mathdist}{dist}
\newcommand{\dist}{{\mathdist}}

\newtheorem{theorem}{Theorem}
\newtheorem{lemma}{Lemma}

\newtheorem{claim}{Claim}

\newtheorem{assumption}{Assumption}

\newcommand{\anonymize}[1]{#1}

\title{Extending Regev's factoring algorithm \\ to compute discrete logarithms}
\author[1,2]{\anonymize{\href{mailto:ekera@kth.se}{Martin Ekerå}}}
\author[1,2]{\anonymize{\href{mailto:jgartner@kth.se}{Joel Gärtner}}}
\affil[1]{\small \anonymize{KTH Royal Institute of Technology, Stockholm, Sweden}}
\affil[2]{\small \anonymize{Swedish NCSA, Swedish Armed Forces, Stockholm, Sweden}}

\predate{\begin{center}}
\postdate{\end{center}\vspace{-3mm}}

\begin{document}
\maketitle

\begin{abstract}
  Regev recently introduced a quantum factoring algorithm that may be perceived as a $d$-dimensional variation of Shor's factoring algorithm.
  In this work, we extend Regev's factoring algorithm to an algorithm for computing discrete logarithms in a natural way.
  Furthermore, we discuss natural extensions of Regev's factoring algorithm to order finding, and to factoring completely via order finding.
  For all of these algorithms, we discuss various practical implementation considerations, including in particular the robustness of the post-processing.
\end{abstract}

\section{Introduction}
\label{sect:introduction}
Regev~\cite{regev23} recently introduced a $d$-dimensional variation of Shor's algorithm~\cite{shor94, shor97} for factoring integers.
The quantum circuit for Regev's algorithm is asymptotically smaller than the circuit for Shor's original algorithm, but the reduction in circuit size comes at the expense of using more space, and of many runs of the circuit being required to achieve a successful factorization.

In this work, we show how Regev's algorithm can be extended to compute discrete logarithms in finite cyclic groups.
In particular, we focus on computing discrete logarithms in~$\mathbb Z_p^*$ for~$p$ a large prime, although the algorithm we present may be generalized to other cyclic groups --- and by extension to Abelian groups.
By comparison, since Regev factors composite~$N$, he implicitly works in~$\mathbb Z_N^*$.

For Regev's algorithm and our extension thereof to reach an advantage over Shor's algorithms~\cite{shor94, shor97}, and the various variations thereof that are in the literature~\cite{ekera-hastad, kaliski, seifert, ekera-pp, ekera-success, ekera-completely, ekera-short-success, ekera-revisiting, ekera-general}, there must exist a notion of \emph{small} elements in the group.

More specifically, there must exist a set of small group elements such that any composition of a subset of these elements under the group operation is also a comparatively small element, and such that small elements are much more efficient to compose than arbitrary elements.

A natural notion of small group elements exists for~$\mathbb Z_p^*$ and~$\mathbb Z_N^*$, for~$p$ a large prime and~$N$ a large composite, respectively.
As such, our extension of Regev's algorithm can compute discrete logarithms in~$\mathbb Z_p^*$ with an asymptotically smaller quantum circuit than the circuits for existing variations of Shor's algorithm.
To the best of our current knowledge, it is however not straightforward to extend the notion of small group elements to elliptic curve groups.

Besides our extension to computing discrete logarithms, we furthermore discuss natural extensions of Regev's algorithm to order finding, and to factoring completely via order finding, in App.~\ref{sect:order-finding-factoring}.

For both Regev's algorithm and our extensions, the asymptotic reduction in the circuit size comes at the expense of using more space, and of having to perform many runs of the circuit.
The outputs from the runs are then jointly post-processed classically.
A practical problem that arises in this context is that the quantum error correction may fail to correct all errors that arise during the course of a run, resulting in an erroneous output being produced.

For this reason, we analyze the robustness of the classical post-processing to errors.
In particular, we heuristically show that it succeeds in recovering the solution even if a significant fraction of the runs produce random outputs.

\section{Preliminaries}
Throughout this work, we write generic Abelian groups multiplicatively, and denote by~$1$ the identity in such groups.
Furthermore, we let~$\rho_s$ be a Gaussian function, defined for $s > 0$ as
\begin{align*}
  \rho_s(\vect{z}) = \exp \left(-\pi \, \frac{\norm{\vect{z}}^2}{s^2} \right)
\end{align*}
where~$\vect z$ is some vector.

\subsection{Recalling Regev's factoring algorithm}
\label{sect:recalling-regevs-factoring-algorithm}
To factor an $n$-bit composite~$N$, Regev~\cite{regev23} defines the lattice
\begin{align*}
  \mathcal L =
    \Bigg\{
      (z_1, \ldots, z_d) \in \mathbb Z^d
      \: \Bigg| \:
      \prod_{i \, = \, 1}^{d} a_i^{z_i} = 1 \:\: (\text{mod } N)
    \Bigg\}
\end{align*}
where $a_i = b_i^2$ for $i \in [1, d] \cap \mathbb Z$, and where $b_1, \ldots, b_d$ are some small $\ordo(\log n)$-bit integers.
In what follows, we assume $b_1, \ldots, b_d$ to be the first $d = \ceil{\sqrt{n}\,}$ primes for simplicity, although other choices are possible, see Sect.~\ref{sect:rationale-choice-generators}.

The quantum part of Regev's algorithm is run $m \ge d + 4$ times to sample a set of~$m$ vectors that can be proved to be close to cosets of~$\mathcal L^*/\mathbb Z^d$.
A classical post-processing algorithm then recovers vectors in~$\mathcal L$ from this set of vectors.

If one of the vectors thus recovered is in fact in~$\mathcal L \backslash \mathcal L^0$, where
\begin{align*}
  \mathcal L^0 =
    \Bigg\{
      (z_1, \ldots, z_d) \in \mathbb Z^d
      \: \Bigg| \:
      \prod_{i \, = \, 1}^{d} b_i^{z_i} \in \{ -1, 1 \} \:\: (\text{mod } N)
    \Bigg\},
\end{align*}
the vector gives rise to a non-trivial split of~$N$.
In practice, our simulations~\cite{eg23-simulations} show that the recovered vectors often give rise to a sufficient number of distinct non-trivial splits so as to yield the complete factorization of~$N$.

In this work, we extend Regev's factoring algorithm to computing discrete logarithms, to order finding, and to factoring completely via order-finding.

\subsubsection{Recalling and adapting Regev's classical post-processing}
To start off, let us recall a claim and a lemma near verbatim from~\cite{regev23} as they form the basis for the classical post-processing in both Regev's algorithm and in the extended algorithms that we introduce in this work:

\begin{claim}[Claim~5.1 from~\cite{regev23}]
  \label{cla:5.1}
  There is an efficient classical algorithm that given a basis of a lattice $\mathcal L \subset \mathbb R^k$, and some norm bound $T > 0$, outputs a list of $\ell \le k$ vectors $\vect{z}_1, \ldots, \vect{z}_{\ell} \in \mathcal L$ of norm at most $2^{k/2} \sqrt{k} \, T$ with the property that any vector in~$\mathcal L$ of norm at most~$T$ must be an integer combination of them.
  In other words, the sublattice they generate contains all the vectors in~$\mathcal L$ of norm at most~$T$.
\end{claim}

The factor~$2^{k/2}$ in the above claim comes from the use of LLL~\cite{lll}.
It can be improved by the use of more powerful lattice reduction algorithms, such as BKZ~\cite{bkz}, but for simplicity we follow Regev and use the above version of the claim based on LLL throughout this work.

\begin{lemma}[Lemma~4.4 from~\cite{regev23}]
  \label{lem:4.4}
  Let~$\mathcal L \subset \mathbb Z^d$ and $m \ge d + 4$.
  Let $\vect{v}_1, \ldots, \vect{v}_m$ be uniformly chosen cosets from~$\mathcal L^* / \mathbb Z^d$.
  For some $\delta > 0$, let $\vect{w}_1, \ldots, \vect{w}_m \in [0,1)^d$ be such that $\dist_{\mathbb R^d/\mathbb Z^d}(\vect w_i, \vect v_i) < \delta$ for all~$i \in [1, d] \cap \mathbb Z$.
  For some scaling factor $S > 0$, define the $d + m$-dimensional lattice~$\mathcal L'$ generated by the rows of~$\vect{B}$ where
  \[\arraycolsep=7pt\def\arraystretch{2}
    \vect{B}
    =
    \left(
    \begin{array}{c|c c c}
      \vect{I}_{d\times d} & S\cdot \vect{w}_1^T & \cdots & S\cdot \vect{w}_m^T \\
      \hline
      \vect{0}_{m \times d}  &  & S \cdot \vect{I}_{m\times m} & \\
    \end{array}
    \right).
  \]

  Then, for any $\vect{u} \in \mathcal L$, there exists a vector $\vect{u}' \in \mathcal L'$ whose first~$d$ coordinates are equal to~$\vect{u}$, and whose norm is at most $\norm{\vect{u}} \cdot (1 + m \cdot S^2 \cdot \delta^2)^{1/2}$.
  Moreover, with probability at least $1/4$ (over the choice of the~$\vect{v}_i$), any non-zero $\vect u' \in \mathcal L'$ of norm $\norm{\vect u'} < \min(S, \delta^{-1}) \cdot \varepsilon/2$ satisfies that its first~$d$ coordinates are a non-zero vector in~$\mathcal L$, where $\varepsilon = (4\det \mathcal L)^{-1/m}/3$.
\end{lemma}

Regev's classical post-processing in~\cite{regev23} essentially follows by combining Lem\-ma~\ref{lem:4.4} and Claim~\ref{cla:5.1} above, as does our classical post-processing.
However, in~\cite{regev23}, Regev only describes the classical post-processing as a part of the main theorem, and not as a separate lemma.
Since we only need the post-processing part of the theorem, we extract this part and give the below lemma that is strongly influenced by Regev's main theorem:

\begin{lemma}[Lemma derived from Theorem~1.1~in~\cite{regev23}]
  \label{lem:post-process}
  Let $m = O(d)$ be an integer not smaller than $d + 4$ and let $T = \exp(\ordo(d))$.
  Furthermore, let~$\mathcal L$ be a $d$-dimensional lattice with $\det \mathcal L < 2^{d^2}$, and let $\vect{v}_1, \ldots, \vect{v}_m$ be uniformly chosen cosets from $\mathcal L^* / \mathbb Z^d$.

  Suppose that we are given a set of~$m$ vectors $\vect{w}_1, \ldots, \vect{w}_m \in [0,1)^d$ such that $\dist_{\mathbb R^d/\mathbb Z^d}(\vect w_i, \vect v_i) < \delta$ for  all $i \in [1, m] \cap \mathbb Z$, where $\delta = \exp(-\ordo(d))$ is sufficiently small.
  Then, there is an efficient classical algorithm that, with probability at least~$1/4$ (over the choice of the~$\vect v_i$), recovers a basis for a sublattice of~$\mathcal L$ that contains all vectors in~$\mathcal L$ of norm at most~$T$.
\end{lemma}
\begin{proof}
  We begin by constructing the lattice~$\mathcal L'$ in Lemma~\ref{lem:4.4} with the given set of vectors $\vect w_1, \ldots, \vect w_m$ and with $S = \delta^{-1}$.
  We then have that any vector $\vect{u} \in \mathcal L$ of norm at most~$T$ corresponds to a vector~$\vect{u}' \in \mathcal L'$ of norm $\norm{\vect{u}'} \le (m+1)^{1/2} T$.

  By using Claim~\ref{cla:5.1} with norm bound $(m+1)^{1/2} T$, we obtain a list of vectors $\vect z'_1, \ldots, \vect z'_\ell$ in $\mathcal L'$ of norm at most
  \[
    2^{(d+m)/2} \cdot (d+m)^{1/2} \cdot (m+1)^{1/2}T < \delta^{-1} (4 \det\mathcal L)^{-1/m}/6
  \]
  where the inequality holds for sufficiently small~$\delta = \exp(-\ordo(d))$ as $\det \mathcal L < 2^{d^2}$.

  From the second property of Lemma~\ref{lem:4.4}, we have that with probability at least~$1/4$, any vector $\vect u' \in \mathcal L'$ of norm $\norm{\vect u'} < \delta^{-1}(4\det \mathcal L)^{-1/m}/6$ satisfies that its first~$d$ coordinates are a vector $\vect u \in \mathcal L$.
  As such, with probability at least $1/4$, the first~$d$ coordinates of~$\vect{z}'_i$ correspond to a vector $\vect{z}_i \in \mathcal L$ for all $i \in [1, \ell] \cap \mathbb Z$.

  Furthermore, any $\vect u \in \mathcal L$ of norm $\norm{\vect u} \le T$ can be written as a linear combination of $\vect z'_1, \ldots, \vect z'_\ell$.
  Thus, $\vect z_1, \ldots, \vect z_\ell$ generate a sublattice of~$\mathcal L$ that contains all vectors $\vect u \in \mathcal L$ of norm $\norm{\vect u} \le T$, and with probability at least $1/4$ the $\ell$ produced vectors can thus be used to recover a basis for this sublattice of $\mathcal L$.
\end{proof}

\subsubsection{Recalling Regev's quantum algorithm}
Regev's algorithm~\cite{regev23} induces an approximation to the state proportional~to
\begin{align*}
  \sum_{\vect{z} \in \{-D/2, \, \ldots, \, D/2-1\}^d}
  \rho_R(\vect{z})
  \ket{z_1, \, \ldots, \, z_{d},
  \prod_{i \, = \, 1}^{d}
  a_i^{z_{i} + D/2} \text{ mod } N}
\end{align*}
where $\vect{z} = (z_1, \ldots, z_d)$.
As for~$R$ and~$D$, Regev first picks $R = 2^{C \sqrt{n}}$ for some suitable constant $C > 0$, and then lets $D = 2^{\lceil \log_2(2 \sqrt{d} \, R) \rceil}$.

Quantum Fourier transforms (QFTs) of size~$D$ are then applied to the first~$d$ control registers, after which the resulting~$d$ frequencies $w_1, \ldots, w_d$ are read out and used to form the vector $\vect{w} = (w_1, \ldots, w_d) / D \in [0, 1)^d$.

The choice of~$C$ impacts the cost of the quantum circuit.
Simulations~\cite{eg23-simulations} show that selecting $C \approx 2$ is sufficient for $n = 2048$ bit integers when $d = \ceil{\sqrt{n}}$.
The probability of successfully factoring the integer in $m = d + 4$ runs is then close to one.
Ragavan and Vaikuntanathan~\cite[App.~A]{rv23} show, based on a heuristic assumption, that selecting $C = 3 + \epsilon + o(1)$ is sufficient asymptotically.

The idea in Regev's factoring algorithm is to leverage that $a_1, \ldots, a_d$ are small by using binary tree-based arithmetic in combination with square--and--multiply-based exponentiation to compute the product $\prod_{i \, = \, 1}^d a_i^{z_i} \text{ mod } N$.\footnote{The constant offset by~$D/2$ is easy to account for by e.g.~instead initializing the work register to a constant.
In practice this offset does not matter, and therefore we do not account for it here.}

More specifically, a work register is first initialized to $1$.
For $j = {\ell-1}, \ldots, 0$ the work register is then squared mod~$N$, the product $\prod_{i \, = \, 1}^d a_i^{z_{i, j}} \text{ mod } N$ computed using binary tree-based arithmetic, and the result multiplied into the work register mod~$N$.
Here, $z_i = \sum_{j \, = \, 0}^{\ell - 1} 2^j z_{i, j}$ where $z_{i, j} \in \{ 0, 1 \}$ and $\ell = \log_2 D$.

An issue with the above approach, as identified by Ekerå and Gidney\footnote{See \href{https://scottaaronson.blog/?p=7489\#comment-1953303}{Craig Gidney's comment on Scott Aaronson's blog}.}, is that the work register cannot be squared in place mod~$N$, leading to intermediary register values having to be kept around until they can be uncomputed at the end of the computation of the product.
This results in an increased space usage.

To circumvent this reversibility issue in Regev's original proposal, Ragavan and Vaikuntanathan~\cite{rv23} have proposed to use Fibonacci-based exponentiation in place of Regev's square--and--multiply-based exponentiation.
This optimization reduces the space requirements --- but it comes at the expense of increasing the circuit size and depth by constant factors.

Even if one accounts for the above optimization, it is currently not clear whether Regev's variation of Shor's algorithm~\cite{shor94, shor97} is more efficient in practice than the various other variations that are in the literature~\cite{ekera-hastad, kaliski, seifert, ekera-pp, ekera-general, ekera-success, ekera-short-success, ekera-completely} in actual physical implementations, and for concrete problem instances of limited size --- depending on what assumptions one makes on the quantum computer, which cost metrics one consider, which problem instances one consider, and so forth.

The same holds true for the extended algorithms that we introduce in this work.
For further discussion on this topic, see Sect.~\ref{sect:practical-efficiency}.

\subsubsection{Extending Regev's quantum algorithm}
\label{sect:extending-quantum}
To extend Regev's algorithm to compute discrete logarithms and orders, we need to introduce at least one element that is not small into the state.

Throughout this work, we therefore consider a slightly extended quantum algo\-rithm that induces an approximation to the state proportional to
\begin{align*}
  \sum_{\vect{z}\in \{-D/2, \ldots, D/2 - 1\}^d}
  \rho_R(\vect{z})
  \ket{
    z_1, \, \ldots, \, z_{d}, \,
    \prod_{i \, = \, 1}^{d-k} g_i^{z_i + D/2}
    \cdot
    \prod_{i \, = \, 1}^{k} u_i^{z_{d-k+i + D/2}}
  }
\end{align*}
where $g_1, \, \ldots, \, g_{d-k}$ are $d-k$ small group elements in some finite Abelian group~$\mathbb G$, and $u_1, \, \ldots, \, u_k$ are~$k$ arbitrary elements in~$\mathbb G$, for~$k$ a small constant.

We focus primarily on the group~$\mathbb Z_N^*$, for~$N$ a prime or composite, although the extended algorithm is generic in that it can also work in other groups.

For $\mathbb G = \mathbb Z_N^*$, we let~$n$ be the bit length of~$N$, where~$n$ serves as an upper bound on both the bit length of the order of~$\mathbb Z_N^*$ and on the bit length of elements in~$\mathbb Z_N^*$.
For other groups, it may be necessary or advantageous to make a distinction between these two bounds.
For the cryptographically relevant cases in~$\mathbb Z_N^*$ this is however not the case, see Sect.~\ref{sect:dlog-pre-compute} for further details.

Furthermore, for simplicity, we let $d = \ceil{\sqrt{n}\,}$, and we let $g_1, \ldots, g_{d-k}$ be the first $d-k$ primes that when perceived as elements of~$\mathbb Z_N^*$ are distinct from $u_1, \ldots, u_k$.
We require~$N$ to be coprime to the first $d$~primes for this reason.

Note however that other choices of~$d$ and $g_1, \ldots, g_{d-k}$ are possible:
In particular, it is possible to let $g_1, \ldots, g_{d-k}$ be any choice of small distinct primes coprime to~$N$, see Sect.~\ref{sect:rationale-choice-generators} for further details.

The idea in our extended algorithms is to use special arithmetic for the first product over $d-k$ small group elements --- i.e.~arithmetic that leverages that these elements are small, in analogy with Regev's original algorithm --- and to use standard arithmetic for the product over the remaining~$k$ elements.

Let us now summarize our analysis in the following lemma:
\begin{lemma}[Extended quantum algorithm based on~\cite{regev23, rv23}]
  \label{lem:quantum-algorithm-Zn}
  Let~$k \ge 0$ be an integer constant, let~$N > 0$ be an $n$-bit integer that is coprime to the first~$d = \ceil{\sqrt{n}\,}$ primes where $d > k$, and let $C > 0$ be a constant.

  Furthermore, let $u_1, \ldots, u_k \in \mathbb Z_N^*$, and let $g_1, \ldots, g_{d-k} \in \mathbb Z_N^*$ be the first $d-k$ primes that when perceived as elements of~$\mathbb Z_N^*$ are distinct from $u_1, \ldots, u_k$.

  Finally, as in~\cite[Theorem~1]{rv23}, let~$G$ be the gate cost of a quantum circuit that takes
  $
  \ket{a, b, t, 0^S} \rightarrow \ket{a, b, (t + ab) \text{ mod } N, 0^S}
  $
  for $a, b, t \in [0, N) \cap \mathbb Z$ and~$S$ the number of ancilla qubits required by the circuit, and let
  \begin{align*}
    \mathcal L
    =
    \Bigg\{
      (z_1, \ldots, z_d) \in \mathbb Z^d
      \: \Bigg| \:
      \prod_{i \, = \, 1}^{d-k} g_i^{z_i}
      \cdot
      \prod_{i \, = \, 1}^{k} u_i^{z_{d-k+i}}
      =
      1
    \Bigg\}.
  \end{align*}

  Then, there is a quantum algorithm that outputs a vector $\vect{w} \in [0, 1)^d$ that, except for with probability $1 / \poly(d)$, is within distance $\sqrt{d/2} \cdot 2^{-C\sqrt{n}}$ of a uniformly chosen coset $\vect{v} \in \mathcal L^*/\mathbb Z^d$.
  The quantum circuit for this algorithm has gate cost $\ordo(n^{1/2} \, G + n^{3/2})$ and it requires
  \begin{align*}
    Q = S + \left( \frac{C}{\log \phi} + 8 + o(1) \right) n
  \end{align*}
  qubits of space, for~$\phi$ the golden ratio.
\end{lemma}
\begin{proof}
  In analogy with Regev~\cite{regev23}, the quantum algorithm first induces an approximation to the state proportional to
  \begin{align}
    \sum_{\vect{z} \in \{-D/2, \ldots, D/2-1\}^d}
    \rho_R(\vect{z})
    &\ket{
      z_1, \, \ldots, \, z_{d}, \,
      \prod_{i \, = \, 1}^{d-k} g_i^{z_i + D/2}
      \cdot
      \prod_{i \, = \, 1}^{k} u_i^{z_{d-k+i} + D/2}
    }
    = \notag \\
    \sum_{\vect{z} \in \{0, \ldots, D-1\}^d}
    \rho_R(\vect{z} - \vect{c})
    &\ket{
      z_1 - D/2, \, \ldots, \, z_{d} - D/2, \,
      \prod_{i \, = \, 1}^{d-k} g_i^{z_i}
      \cdot
      \prod_{i \, = \, 1}^{k} u_i^{z_{d-k+i}}
    } \label{eq:psi-extended}
  \end{align}
  where $\vect{z} = (z_1, \ldots, z_d)$ and $\vect{c} = (D/2, \ldots, D/2)$.

  As for the parameters~$R$ and~$D$, in analogy with Regev~\cite{regev23}, we pick ${R = 2^{C \sqrt{n}}}$ for some suitable constant~$C > 0$, and let $D = 2^{\lceil \log_2(2 \sqrt{d} \, R) \rceil}$.
  QFTs of size~$D$ are then applied to the first~$d$ control registers $z_1, \ldots, z_d$ and the resulting frequencies $w_1, \ldots, w_d$ read out, yielding $\vect{w} = (w_1, \ldots, w_d) / D \in [0, 1)^d$.
  Based on the analysis in~\cite{regev23}, except for with probability $1 / \poly(d)$, the vector~$\vect{w}$ is then within distance $\sqrt{d}/(\sqrt{2}R) = \sqrt{d/2} \cdot 2^{-C\sqrt{n}}$ of a uniformly chosen coset $\vect{v} \in \mathcal L^*/\mathbb Z^d$.

  To simplify the notation in what follows, we shift the first~$d$ control registers by~$D/2$ since this does not affect the results measured after the QFTs have been applied.
  That is to say, instead of inducing an approximation to the state proportional to~$\refeq{psi-extended}$, we induce an approximation to the state proportional to
  \begin{align*}
    \sum_{\vect{z} \in \{0, \ldots, D-1\}^d}
    \rho_R(\vect{z} - \vect{c})
    \ket{
      z_1, \, \ldots, \, z_{d}, \,
      \prod_{i \, = \, 1}^{d-k} g_i^{z_i}
      \cdot
      \prod_{i \, = \, 1}^{k} u_i^{z_{d-k+i}}
    }
  \end{align*}
  in a series of computational steps.

  The first step is to approximate the state proportional to
  \begin{align*}
      \sum_{\vect{z}\in \{0, \ldots, D-1\}^d} \rho_R(\vect z-\vect{c}) \ket{z_1, \ldots, z_d}
  \end{align*}
  to within $1 / \poly(d)$ via a quantum circuit of size $d(\log D + \poly(\log d))$ as explained in~\cite{regev23, rv23} by referencing~\cite{regev09}.
  Note that $\ordo(d(\log D + \poly(\log d))) = \ordo(n)$.

  To compute the first product, i.e.\ to take
  \begin{align*}
    &\sum_{\vect{z}\in \{0, \ldots, D-1\}^d} \rho_R(\vect z-\vect{c}) \ket{z_1, \ldots, z_d, 0^A}
    \\
    \rightarrow
    &\sum_{\vect{z}\in \{0, \ldots, D-1\}^d} \rho_R(\vect z-\vect{c}) \ket{z_1, \ldots, z_d, \prod_{i \, = \, 1}^{d-k} g_i^{z_i}, 0^{A-n}}
  \end{align*}
  a special quantum circuit from~\cite{rv23} is used, that for efficiency reasons leverages that the group elements $g_1, \ldots, g_{d-k}$ are small.

  By~\cite[Lemma~2.2]{rv23}, this circuit uses $\ordo(n^{1/2} \, G + n^{3/2})$ gates and $Q$~qubits of space\footnote{Note that the constant~$C$ in~\cite{rv23} is different from our constant $C$: Whereas we have defined~$C$ so that $R = 2^{Cd}$ as in~\cite{regev23}, the authors of~\cite{rv23} have instead defined~$C$ so that $R = 2^{(C+2+o(1))d}$.
  This explains the difference between our expression for~$Q$ and the corresponding expression in~\cite{rv23}.} in total when accounting both for the~$d$ control registers and the~$A$ initial ancilla qubits, $n$~qubits of which are used to store the product.

  Note that the circuit in~\cite[Lemma~2.2]{rv23} is for the product of the squares of the first~$d$ primes raised to short exponents.
  We compute the product of a subset of $d-k$ of the first $d$ primes raised to short exponents, but these differences only serve to reduce the cost of the circuit.

  To compute the second product, i.e.\ to take
  \begin{align*}
    &\sum_{\vect{z} \in \{0, \ldots, D-1\}^d} \rho_R(\vect z-\vect{c}) \ket{z_1, \ldots, z_d, \prod_{i \, = \, 1}^{d-k} g_i^{z_i}, 0^{A-n}} \\
    \rightarrow
    &\sum_{\vect{z} \in \{0, \ldots, D-1\}^d} \rho_R(\vect z-\vect{c}) \ket{z_1, \ldots, z_d, \prod_{i \, = \, 1}^{d-k} g_i^{z_i} \cdot \prod_{i \, = \, 1}^{k} u_i^{z_{d-k+i}}, 0^{A-n}}
  \end{align*}
  a standard quantum circuit is used since~$u_1, \ldots, u_k$ are not necessarily small.

  More specifically, $2 k \log D$ elements are classically pre-computed, and then multiplied and added into the work register conditioned on the control qubits:

  Recall that~$D$ is a power of two, write $D = 2^\ell$, write
  \begin{align*}
    z_{d-k+i} = \sum_{j \, = \, 0}^{\ell - 1} 2^j z_{d-k+i, j}
    \: \text{ for } \:
    z_{d-k+i, j} \in \{0, 1\}
    \quad \Rightarrow \quad
    u_i^{z_{d-k+i}} = \prod_{j \, = \, 0}^{\ell - 1} u_i^{2^j z_{d-k+i, j}},
  \end{align*}
  and pre-compute $u_i^{2^j}$ and $u_i^{-2^j}$ for all $j \in [0, \ell) \cap \mathbb Z$ and $i \in [1, k] \cap \mathbb Z$.
  Then multiply~$u_i^{2^j}$ into the work register conditioned on~$z_{d-k+i, j}$ for all $j \in [0, \ell) \cap \mathbb Z$ and $i \in [1, k] \cap \mathbb Z$.
  For~$p$ the product in the work register before multiplying in~$u_i^{2^j}$, and
  \begin{align*}
    u
    &=
    \left\{\def\arraystretch{1.3}
    \begin{array}{ccc}
      u_i^{2^j} & \text{ if } & z_{d-k+i, j} = 1 \\
      1         & \text{ if } & z_{d-k+i, j} = 0
    \end{array}
    \right.
  \end{align*}
  the procedure, when perceived to work in~$\mathbb Z_N$, takes
  \begin{align*}
    &\ket{z_1, \ldots, z_d, p, 0^n, 0^n, 0^{A-3n}} \\
    \text{load u} \rightarrow
    &\ket{z_1, \ldots, z_d, p, u, 0^n, 0^{A-3n}} \\
    \text{multiply and add} \rightarrow
    &\ket{z_1, \ldots, z_d, p, u, pu, 0^{A-3n}} \\
    \text{swap} \rightarrow
    &\ket{z_1, \ldots, z_d, pu, u, p, 0^{A-3n}} \\
    \text{unload~$u$ and load $-u^{-1}$} \rightarrow
    &\ket{z_1, \ldots, z_d, pu, -u^{-1}, p, 0^{A-3n}} \\
    \text{multiply and add} \rightarrow
    &\ket{z_1, \ldots, z_d, pu, -u^{-1}, 0^n = p + pu \cdot -u^{-1}, 0^{A-3n}} \\
    \text{unload $-u^{-1}$} \rightarrow
    &\ket{z_1, \ldots, z_d, pu, 0^n, 0^n, 0^{A-3n}}
  \end{align*}
  where the load and unload operations are conditioned on $z_{d-k+i, j}$.

  The~$S$ qubits required to perform the multiplications and additions fit into the $A - 3n$ ancilla qubits.
  The circuit size is dominated by the $2 k \log D$ multiplications and additions, with a cost of $\ordo(2 k G \log D) = \ordo(n^{1/2} \, G)$ gates.
  The loading, unloading and swap operations may be modelled as requiring~$\ordo(n)$ gates.
  Hence the second product is on par with or less expensive than the first asymptotically.

  Finally, as explained in~\cite{regev23, rv23}, the QFTs of the first~$d$ control registers may be approximated in-place by a circuit of size $\ordo(d \log D(\log \log D + \log d)) = \ordo(n \log n)$ to within $1 / \poly(d)$ with the algorithm of Coppersmith~\cite{coppersmith2002}.
\end{proof}

By plugging concrete expression for~$G$ and~$S$ into Lemma~\ref{lem:quantum-algorithm-Zn} above for various circuits that implement the required arithmetic, we can obtain corollaries that give concrete circuit costs, in analogy with~\cite[Corollary~1.1--1.3 to Theorem~1]{rv23}.

Lemma~\ref{lem:quantum-algorithm-Zn} may furthermore be extended to other groups for which the group operation may be implemented efficiently quantumly.
Note however that it is of course not possible to give asymptotic costs in terms of circuit size and space usage unless the group is explicitly specified.

Note furthermore that for the algorithm to be more efficient than other variations of Shor's algorithm the implementation needs to leverage that all but~$k$ of the~$d$ elements in the product are small.

\subsection{Notes on lattice determinants}
The post-processing in Lemma~\ref{lem:post-process} requires that the lattice~$\mathcal L$ has determinant less than~$2^{d^2}$.
In Regev's original factoring algorithm~\cite{regev23}, the lattice considered is
\begin{align*}
  \mathcal L
  =
  \Bigg\{
    (z_1, \ldots, z_d) \in \mathbb Z^d
    \: \Bigg| \:
    \prod_{i \, = \, 1}^{d} a_i^{z_i} = 1 \:\: (\text{mod } N)
  \Bigg\},
\end{align*}
and Regev shows that its determinant is less than $N \le 2^{d^2}$.

For our extensions of Regev's algorithm, we prove the below, more general, lemma to bound the determinant of lattices of this form.
\begin{lemma}
  \label{lem:determinant}
  Let $g_1, \ldots, g_t$ be $t > 0$ elements of a finite Abelian group~$\mathbb G$, and let
  \begin{align*}
    \mathcal L
    =
    \Bigg\{
      (z_1, \ldots, z_t) \in \mathbb Z^t
      \: \Bigg| \:
      \prod_{i \, = \, 1}^{t} g_i^{z_i} = 1
    \Bigg\}.
  \end{align*}

  Then, the determinant of~$\mathcal L$ is equal to the size of the subgroup of~$\mathbb G$ that is generated by the elements $g_1, \ldots, g_t$.
\end{lemma}
\begin{proof}
  Let~$\mathbb G_0$ be the subgroup of~$\mathbb G$ that is generated by $g_1, \ldots, g_t$.
  For each group element $g \in \mathbb G_0$, we define the coset~$\mathcal S_g$ of $\mathbb Z^t/\mathcal L$ given by
  \[
    \mathcal S_g
    =
    \Bigg\{
      (z_1, \ldots, z_t) \in \mathbb Z^t
      \: \Bigg| \:
      \prod_{i \, = \, 1}^{t} g_i^{z_i} = g
    \Bigg\}.
  \]

  We obviously have that these cosets are distinct, and that every $\vect{z} \in \mathbb Z^t$ belongs to one of these cosets.
  Since~$\mathcal L$ is an integer lattice, the determinant of~$\mathcal L$ is given by the number of such distinct cosets of~$\mathbb Z^t/\mathcal L$.
  Since each element of~$\mathbb G_0$ corresponds to a unique coset~$\mathcal S_g$, the determinant of~$\mathcal L$ therefore equals the size of the subgroup~$\mathbb G_0$ generated by the elements $g_1, \ldots, g_t \in \mathbb G$.
\end{proof}

In particular, if the group~$\mathbb G$ is of size less than~$2^n$, the above lemma guarantees that, no matter the choice of elements $g_1, \ldots, g_t$, the corresponding lattice has determinant less than $2^n \le 2^{d^2}$ when $d = \ceil{\sqrt{n}\,}$.

\subsection{Notes on the choice of small generators}
\label{sect:rationale-choice-generators}
Both Regev's original factoring algorithm and our extended algorithms assume that there exist short \emph{interesting} vectors in a lattice
\[
  \mathcal L
  =
  \Bigg\{
    (z_1, \ldots, z_d) \in \mathbb Z^d
    \: \Bigg| \:
    \prod_{i \, = \, 1}^{d} g_i^{z_i} = 1
  \Bigg\}
\]
for some choice of generators~$g_1, \ldots, g_d \in \mathbb G$, most or all of which are small.

What constitutes an interesting vector differs between our algorithms and Regev's algorithm, but in both cases it is sufficient that~$\mathcal L$ has a short basis.\footnote{For Regev's algorithm, a short basis for $\mathcal L$ ensures that there are short vectors in $\mathcal L$ that are not in $\mathcal L^0 \subsetneq \mathcal L$, for $\mathcal L^0$ as in Sect.~\ref{sect:recalling-regevs-factoring-algorithm}.}

In turn, whether or not~$\mathcal L$ has a short basis depends on the specific choice of generators.
To give an example of a bad choice of generators, we can let $g_i = g^i$ for~$g$ a small generator of order~$r$.
In this case,
\begin{align*}
  \mathcal L
  &=
  \Bigg\{
    (z_1, \ldots, z_d) \in \mathbb Z^d
    \: \Bigg| \:
    \prod_{i \, = \, 1}^{d} g^{i\cdot z_i} = 1
  \Bigg\} \\
  &=
  \Bigg\{
    (z_1, \ldots, z_d) \in \mathbb Z^d
    \: \Bigg| \:
  \sum_{i \, = \, 1}^{d} i \cdot z_i = 0 \:\: (\text{mod } r)
  \Bigg\},
\end{align*}
and any basis of~$\mathcal L$ would have to contain a vector~$\vect{z} = (z_1, \ldots, z_d)$ such that the sum over $i \cdot z_i$ is equal to a non-zero multiple of~$r$.
For such~$\vect{z}$, we thus have that
\[
  r \le \abs{\, \sum_{i \, = \, 1}^d i \cdot z_i \,}
  \le
  d \sum_{i\, = \, 1}^d \abs{z_i}
  \le
  d^{3/2} \norm{\vect{z}}
\]
and hence, the basis would have to contain a vector of length~$\norm{\vect{z}} \ge r / d^{3/2}$, which is too large in this context for the basis to be considered short.

If the generators are chosen to be void of intentionally introduced relations, it is natural to assume~$\mathcal L$ to have some properties similar to those of random lattices.
In particular, it is natural to assume that there is a basis of~$\mathcal L$ where each vector is not significantly longer than the shortest non-zero vector in $\mathcal L$.
By Minkowski's first theorem and the fact that $\det \mathcal L \le 2^n$, we thus expect~$\mathcal L$ to have a basis where each basis vector has length at most $\exp(\ordo(n/d))$.

Note however that we do not prove that such a basis exists for a given choice of generators.
Rather, we make a heuristic number theoretic assumption that such a basis exists, in analogy with Regev~\cite{regev23}.

For $\mathbb G = \mathbb Z_N^*$, for~$N$ prime or composite, a natural choice to avoid intentionally introduced relations is to let most of the generators $g_1, \ldots, g_d$ be small distinct primes coprime to~$N$.
Our concrete assumption is therefore the following:
\begin{assumption} \label{assump:basis-generic}
  Let~$K$ be a constant, $N$ be an~$n$-bit integer coprime to the first $d = O(\sqrt{n})$ primes, and $g_1, \ldots, g_{d}$ be the first $d$ primes perceived as elements in~$\mathbb Z_N^*$.
  Then, $g_1, \ldots, g_d$ span $\mathbb Z_N^*$ and the lattice
  \[
    \Bigg\{
    (z_1, \ldots, z_d) \in \mathbb Z^{d}
      \: \Bigg| \:
      \prod_{i \, = \, 1}^{d} g_i^{z_i} = 1
    \Bigg\}
  \]
  has a basis where each basis vector has norm at most $T = \exp(Kn/d)$.
\end{assumption}

Note that we could instead make a more flexible assumption, by letting the generators $g_1, \ldots, g_{d}$ be small distinct primes, but not necessarily the first $d$ primes.
For our algorithm to work it would then be sufficient for the assumption to hold with a noticeable probability over the specific choice of these generators, as we could re-run the algorithm with different choices of generators.

Furthermore, note that the above assumption predicts that the lattice has increasingly shorter bases as $d$ grows larger.
As such, choosing a larger value for $d$ may be preferable in practice.
This being said, for simplicity we consider only the special choice of using the first $d = \ceil{\sqrt{n}}$ primes as generators in this paper.

Finally, note that our extensions of Regev's algorithm do not directly depend on the lattice in Assumption~\ref{assump:basis-generic}.
Instead, they require a closely related lattice to have a short basis.
That this is the case does however easily follow from Assumption~\ref{assump:basis-generic}, as detailed in the following lemma:
\begin{lemma} \label{lem:basis-dlog}
  Let~$N$ be an $n$-bit integer and $g_1, \ldots, g_{d}$ be the first $d = O(\sqrt{n})$ primes
  perceived as elements in~$\mathbb Z_N^*$.
  Furthermore, let~$K$ be a constant,~$k$ be some small constant, and $u_1, \ldots, g_k$ be arbitrary elements in $\mathbb Z_N^*$.
  Then, under Assumption~\ref{assump:basis-generic}, the lattice
  \[
    \Bigg\{
      (z_1, \ldots, z_{d+k}) \in \mathbb Z^{d+k}
      \: \Bigg| \:
      \prod_{i \, = \, 1}^{d} g_i^{z_i} \prod_{i \, = \, 1}^{k} u_i^{z_{d+i}} = 1
    \Bigg\}
  \]
  has a basis where each basis vector has norm at most $T = \sqrt{d+1} \cdot \exp(Kn/d)$.
\end{lemma}
\begin{proof}
  Consider the lattice $\mathcal L$ given by
  \[
    \Bigg\{
    (z_1, \ldots, z_d) \in \mathbb Z^{d}
      \: \Bigg| \:
      \prod_{i \, = \, 1}^{d} g_i^{z_i} = 1
    \Bigg\}
  \]
  that, under Assumption~\ref{assump:basis-generic}, has a basis~$\vect{B}$ where each basis vector has norm at most $\exp(Kn/d)$.
  By Assumption~\ref{assump:basis-generic} we also have that $g_1, \ldots, g_d$ generate $\mathbb Z_N^*$.

  Therefore, for any of the $u_i$, there is a vector $\vect{z} = (z_1, \ldots, z_d) \in \mathbb Z^d$ such that $u_i = \prod_{i \, = \, 1}^{d} g_i^{z_i}$.
  As the longest vector in $\vect B$ is no longer than $\exp(Kn/d)$, the covering radius of $\mathcal L$ is at most $\sqrt{d}\exp(Kn/d)$.
  As such, $\vect{z}$ can be written as a vector in $\mathcal L$ plus a vector~$\vect z_{u_i}$ of length at most $\sqrt{d}\exp(Kn/d)$.

  We then see that a basis for the lattice
  \[
    \Bigg\{
      (z_1, \ldots, z_{d+k}) \in \mathbb Z^{d+k}
      \: \Bigg| \:
      \prod_{i \, = \, 1}^{d} g_i^{z_i} \prod_{i \, = \, 1}^{k} u_i^{z_{d+i}} = 1
    \Bigg\}
  \]
  is given by
  \[\arraycolsep=7pt\def\arraystretch{2}
  \left(
  \begin{array}{c|c}
    \vect{B} & \vect{0}_{d\times k} \\
    \hline
    \vect{z}_{u_1} & \\
    \vdots & -\vect{I}_{k} \\
    \vect{z}_{u_k} &
  \end{array}
  \right)
  \]
  where each basis vector has length at most $\sqrt{d+1} \cdot \exp(Kn/d)$.
\end{proof}

\section{Computing discrete logarithms}
In what follows, let~$\mathbb G$ be a finite Abelian group, let~$g \in \mathbb G$ be a generator of a cyclic group $\langle g \rangle \subseteq \mathbb G$ of order~$r$, and let $x = g^e$ for~$e \in (0, r) \cap \mathbb Z$.
Given~$g$ and~$x$, our goal is then to compute the discrete logarithm $e = \log_g x$.

To this end, we first give a basic algorithm in Sect.~\ref{sect:dlog-pre-compute} that computes the discrete logarithm whilst requiring pre-computation for the group.
We then show in Sect.~\ref{sect:dlog-inc-compute} how the pre-computation can be performed as a part of the algorithm.

Both variants work if Assumption~\ref{assump:basis-generic} holds, but a weaker heuristic assumption is sufficient for the algorithm with pre-computation to work.
We have heuristically verified that both variants work in practice by means of simulations.

\subsection{A basic algorithm that requires pre-computation}
\label{sect:dlog-pre-compute}
Let $g_1, \, \ldots, \, g_{d-1}$ be small distinct elements in $\langle g \rangle$, such that $g_i = g^{e_i}$ where $e_i \in (0, r) \cap \mathbb Z$ for all $i \in [1, d) \cap \mathbb Z$.
In particular, for $g \in \mathbb Z_p^*$ for~$p$ a large $n$-bit prime, and for $d = \ceil{\sqrt{n}\,}$, we take $g_1, \ldots, g_{d-1}$ to be the first $d-1$ primes that are in~$\langle g \rangle$ and that when perceived as elements of~$\langle g \rangle$ and distinct from~$x$.

To find such elements efficiently, and to be able to directly call upon Lemma~\ref{lem:quantum-algorithm-Zn} below, we assume in this section that $\langle g \rangle = \mathbb Z_p^*$.
Note that this assumption does not imply a loss of generality:
Suppose that $\gamma = g^{e_\gamma}$ generates some subgroup of~$\mathbb Z_p^*$ of order~$r_\gamma$, that $x = \gamma^{e_x} = g^{e}$ and that we seek $e_x = \log_\gamma x \in [0, r_\gamma) \cap \mathbb Z$.
Then, we may compute $e = \log_g x$ and $e_\gamma = \log_g \gamma$, and finally $e_x = e / e_\gamma \text{ mod } r_\gamma$.

By Lemma~\ref{lem:quantum-algorithm-Zn}, there exists an efficient quantum algorithm that outputs a vector $\vect{w} \in [0, 1)^d$ that, except for with probability $1 / \poly(d)$, is within distance $\sqrt{d/2} \cdot 2^{-C\sqrt{n}}$ of a uniformly chosen $\vect{v} \in \mathcal L_{x}^*/\mathbb Z^d$, where
\begin{align*}
  \mathcal L_x
  =
  \Bigg\{
    (z_1, \ldots, z_d) \in \mathbb Z^d
    \: \Bigg| \:
    x^{z_d} \prod_{i \, = \, 1}^{d-1} g_i^{z_i} = 1
  \Bigg\}
\end{align*}
for $C > 0$ some constant.
By performing $m \ge d + 4$ runs of this quantum algorithm, we obtain~$m$ such vectors $\vect{w}_1, \ldots, \vect{w}_m$.

By Lemma~\ref{lem:post-process}, there is an efficient classical algorithm that, with probability at least $1/4$, recovers a basis for a sublattice~$\Lambda$ of~$\mathcal L_x$ given the vectors $\vect{w}_1, \ldots, \vect{w}_m$.
Furthermore, $\Lambda$ contains all vectors in~$\mathcal L_x$ of norm at most~${T = \exp(\ordo(\sqrt{n}))}$.
Let
\begin{align*}
  \mathcal L_x^0
  =
  \Bigg\{
    (z_1, \ldots, z_d) \in \mathcal L_x
    \: \Bigg| \:
    \gcd(z_d, r) \neq 1
  \Bigg\} \subset \mathcal L_x
\end{align*}
and note that $\mathcal L_x^0$ is a sublattice of~$\mathcal L_x$.

We make the heuristic assumption that there is a non-zero vector ${\vect{y} \in \mathcal L_x \backslash \mathcal L_x^0}$ of norm at most $T$.
It then follows that $\vect{y} \in \Lambda$.
By extension, it follows that at least one vector $\vect{z} \in \mathcal S$ must also be in~$\mathcal L_x \backslash \mathcal L_x^0$, as the sublattice~$\Lambda$ that the vectors in~$\mathcal S$ generate would otherwise not include~$\vect{y}$.

Suppose that we know~$e_i$ such that $g_i = g^{e_i}$ for $i \in [1, d) \cap \mathbb Z$ --- e.g.~because we have pre-computed these using Shor's algorithm for computing discrete logarithms~\cite{shor94, shor97, ekera-revisiting}.
Then, given $\vect{z} = (z_1, \ldots, z_d) \in \mathcal L_x\backslash \mathcal L_x^0$, we have that
\begin{align*}
  x^{z_d} \prod_{i \, = \, 1}^{d-1} g_i^{z_i}
  =
  g^{e_1 z_1 + \ldots + e_{d-1} z_{d-1} + e z_d}
  =
  g^0
  =
  1,
\end{align*}
which implies that
\begin{align*}
  e_1 z_1 + \ldots + e_{d-1} z_{d-1} + e z_d = 0 \quad (\text{mod } r).
\end{align*}

Since $\vect{z} \in \mathcal L_x \backslash \mathcal L_x^0$, we have that $\gcd(z_d, r) = 1$, so~$z_d$ is invertible mod~$r$, and hence we can compute the discrete logarithm
\begin{align*}
  e = \log_g x = -(e_1 z_1 + \ldots + e_{d-1} z_{d-1}) / z_d \quad (\text{mod } r).
\end{align*}

We have heuristically verified that the above procedure yields a vector in $\mathcal L_x \backslash \mathcal L_x^0$, and that we can solve for the logarithm~$e$, by means of simulations.
For $n = 2048$ bits, the simulations indicate that it suffices to take $C \approx 2$, and that the success probability is close to one after $d + 4$ runs of the quantum algorithm.

\subsubsection{Notes on safe-prime groups and Schnorr groups}
\label{sect:notes-safe-prime-and-schnorr-groups}
In cryptographic applications of the discrete logarithm problem in~$\mathbb Z_p^*$, either safe-prime groups or Schnorr groups are typically used in practice:

For safe-prime groups, we have that $p - 1 = 2 r_\gamma$ for~$r_\gamma$ a prime, where~$\gamma$ generates the $r_\gamma$-order subgroup.
In such groups, it is easy to find small~$g_i$ in~$\langle \gamma \rangle$.
Hence, it would be possible to modify the quantum algorithm in the previous section to work in~$\langle \gamma \rangle$ instead of working in all of~$\mathbb Z_p^*$ for safe-prime groups.

This being said, the orders of~$\langle \gamma \rangle$ and~$\mathbb Z_p^*$ differ only by a factor of two, so the advantage of working in $\langle \gamma \rangle$ is small.
Furthermore, making the aforementioned modification would require us to design a special version of Lemma~\ref{lem:quantum-algorithm-Zn} to e.g.~use separate bounds for the bit length of group elements and of the group order, respectively.
We therefore only give the basic algorithm that works in~$\mathbb Z_p^*$.

For Schnorr groups, we have that $p - 1 = 2 k r_\gamma$ for~$r_\gamma$ a prime and~$k$ a large integer, where~$\gamma$ generates the $r_\gamma$-order subgroup.
In such groups, it is typically hard to find small~$g_i$ that are in~$\langle \gamma \rangle$.
To overcome this problem, we need to work in a larger group, again leaving us in a similar situation to that for safe-prime groups:
We could for instance work in the $k r_\gamma$-order subgroup, but this would only bring a small advantage compared to working in~$\mathbb Z_p^*$, so we only give the basic algorithm that works in~$\mathbb Z_p^*$.

In summary, the complexity of the quantum algorithm in the previous section depends exclusively on the bit length~$n$ of~$p$, irrespective of whether we are in a Schnorr group or a safe-prime group, and in the latter case irrespective of whether the logarithm is short or full length.
For Shor's algorithm, when solving $x = \gamma^{e_x}$ for~$e_x$, the complexity depends on~$n$ and the bit length of~$r_\gamma$ when the algorithm is adapted as in~\cite{ekera-revisiting, ekera-general}, or on~$n$ and the bit length of~$e_x$ when it is adapted as in~\cite{ekera-hastad, ekera-pp, ekera-short-success}.

\subsubsection{Notes on other options for the pre-computation}
Another way to perform the pre-computation is to use an analogous procedure to compute a basis of~$\mathcal L_{g}$, and to use it to compute $\vect{e} = (e_1, \, \ldots, \, e_{d-1})$, where
\begin{align*}
  \mathcal L_g
  =
  \Bigg\{
    (z_1, \ldots, z_d) \in \mathbb Z^d
    \: \Bigg| \:
    g^{z_d} \prod_{i \, = \, 1}^{d-1} g_i^{z_i} = 1
  \Bigg\}.
\end{align*}

This requires a stronger heuristic assumption to be made, on the existence of a basis of~$\mathcal L_{g}$ with all basis vectors having norm at most $T = \exp(\ordo(\sqrt{n}))$.
In the next section we integrate this pre-computation into the algorithm.

\subsection{Integrating the pre-computation into the algorithm}
\label{sect:dlog-inc-compute}
In the algorithm in Sect.~\ref{sect:dlog-pre-compute}, we have to pre-compute~$e_i \in (0, r) \cap \mathbb Z$ such that $g_i = g^{e_i}$ for all $i \in [1, d) \cap \mathbb Z$.
To avoid this pre-computation, we can define a lattice~$\mathcal L_{x,g}$ that depends on both~$g \in \mathbb G$ of order $r$ and~$x = g^e$ for $e \in (0, r) \cap \mathbb Z$, alongside $d - 2$ small elements $g_1, \ldots, g_{d-2} \in \mathbb G$.
In particular, for~$g \in \mathbb Z_p^* = \mathbb G$ for~$p$ a large $n$-bit prime and $d = \ceil{\sqrt{n}\,}$, we take $g_1, \ldots, g_{d-2}$ to be the first $d-2$ primes that when perceived as elements of~$\mathbb Z_p^*$ are distinct from~$g$ and~$x$.

By Lemma~\ref{lem:quantum-algorithm-Zn}, there exists an efficient quantum algorithm that outputs a vector $\vect{w} \in [0, 1)^d$ that, except for with probability $1 / \poly(d)$, is within distance ${\sqrt{d/2} \cdot 2^{-C\sqrt{n}}}$ of a uniformly chosen $\vect{v} \in \mathcal L_{x,g}^*/\mathbb Z^d$, where
\begin{align*}
  \mathcal L_{x,g}
  =
  \Bigg\{
  (z_1, \ldots, z_d) \in \mathbb Z^{d}
    \: \Bigg| \:
    x^{z_{d-1}} g^{z_{d}} \prod_{i \, = \, 1}^{d-2} g_i^{z_i} = 1
  \Bigg\},
\end{align*}
for $C > 0$ some constant.
By performing $m \ge d + 4$ runs of this quantum algorithm, we obtain~$m$ such vectors $\vect{w}_1, \ldots, \vect{w}_m$.

By Lemma~\ref{lem:post-process}, there is an efficient classical algorithm that, with probability at least $1/4$, recovers a basis for a sublattice~$\Lambda$ of~$\mathcal L_{x, g}$ given the vectors $\vect{w}_1, \ldots, \vect{w}_m$.
Furthermore,~$\Lambda$ contains all vectors in~$\mathcal L_{x, g}$ of norm at most~$T = \exp(\ordo(\sqrt n))$.

Lemma~\ref{lem:basis-dlog} gives that, under Assumption~\ref{assump:basis-generic}, the lattice $\mathcal L_{x, g}$ has a basis where each basis vector has norm at most $\exp(\ordo(n/d)) = \exp(\ordo(\sqrt{n})) = T$.
Hence $\Lambda$ is the full lattice $\mathcal L_{x, g}$, and Lemma~\ref{lem:post-process} recovers a basis~$\vect{B}$ for this lattice.
By the definition of~$\mathcal L_{x,g}$, we have that ${\vect{v}_e = (0, \ldots, 0, 1, -e) \in \mathcal L_{x,g}}$ as ${x^1 g^{-e} = 1}$.
Since we know a basis~$\vect B$ for~$\mathcal L_{x,g}$, we can efficiently recover~$\vect{v}_e$ and by extension~$e$.

Note furthermore that $\vect{v}_r = (0, \ldots, 0, r) \in \mathcal L_{x,g}$, and that we can efficiently recover~$\vect{v}_r$ and hence~$r$.
Hence, if~$r$ is unknown, we can recover~$r$ along with~$e$ at no additional quantum cost.
This is useful in Sect.~\ref{sect:reduce-circuit-size} below where~$r$ is needed.

We have heuristically verified that the above procedure yields a basis of~$\mathcal L_{x,g}$, and that we can solve for the logarithm~$e$ and order~$r$, by means of simulations.
For $n = 2048$ bits, the simulations indicate that it suffices to take $C \approx 2$, and that the success probability is close to~$1$ after $d + 4$ runs of the quantum algorithm.

We are now ready to summarize the above analysis in a theorem:
\begin{theorem}
  \label{thm:main}
  Let $p$ be an $n$-bit prime, and let $d = \ceil{\sqrt{n} \,}$.
  Let $g \in \mathbb Z_p^*$ be of order~$r$, and let $x = g^e$ for $e \in (0, r) \cap \mathbb Z$.
  As in~\cite[Theorem~1]{rv23}, let~$G$ be the gate cost of a quantum circuit that takes
  \begin{align*}
    \ket{a, b, t, 0^S} \rightarrow \ket{a, b, (t + ab) \text{ mod } p, 0^S}
  \end{align*}
  for $a, b, t \in [0, p) \cap \mathbb Z$ and~$S$ the number of ancilla qubits required.
  Then, under Assumption~\ref{assump:basis-generic}, there is an efficient classical algorithm that, by calling a quantum circuit $d + 4$ times, yields the logarithm~$e$ with probability at least~$1/4$.
  This quantum circuit has gate cost $\ordo(n^{1/2} \, G + n^{3/2})$, and it requires
  \begin{align*}
    S + \left( \frac{C}{\log \phi} + 8 + o(1) \right) n
  \end{align*}
  qubits of space, for some constant $C > 0$ and~$\phi$ the golden ratio.
\end{theorem}
\begin{proof}
  The proof follows from the above analysis, and from Lemma~\ref{lem:quantum-algorithm-Zn} where the gate and space costs of the quantum circuit are analyzed.
\end{proof}

The above theorem is specific to~$\mathbb Z_p^*$ for~$p$ prime since there is a notion of small group elements in such groups, and since such groups are extensively used in cryptography.
As previously stated, the algorithm may be generalized to other Abelian groups, but for the algorithm to have an advantage over other algorithms in the literature there must exist a notion of small elements in the group.

\subsection{Notes on reducing the circuit size by pre-computation}
\label{sect:reduce-circuit-size}
In contrast to the algorithm in Sect.~\ref{sect:dlog-pre-compute}, the algorithm in Sect.~\ref{sect:dlog-inc-compute} does not require any pre-computation, and it does not require $g$ to generate~$\mathbb G$.

These benefits come with the downside, however, of the quantum part of the algorithm having to exponentiate both~$x$ and~$g$, where neither element is guaranteed to be small.
In practice, these exponentiations constitutes a significant fraction of the cost of each run of the quantum part of the algorithm.
In turn, this makes the algorithm in Sect.~\ref{sect:dlog-inc-compute} somewhat less efficient than the algorithm in Sect.~\ref{sect:dlog-pre-compute} that only has to exponentiate~$x$ at the expense of pre-computing $e_1, \ldots, e_{d-1}$ with respect to~$g$.

One way to overcome this issue is to let $g'$ be a small element in $\mathbb G$ such that $g \in \langle g' \rangle \subseteq \mathbb G$, and to first pre-compute $e_g = \log_{g'} g$ quantumly.
It then suffices to compute $e_x = \log_{g'} x$ quantumly for each~$x$, and to return $e = e_x / e_g \text{ mod } r$.
All but one of the~$d$ elements that are exponentiated in each run of the quantum part of the algorithm are then small as desired.
Both the pre-computation of~$e_g$ and the computation of~$e_x$ for different~$x$ may be performed with the quantum algorithm in Sect.~\ref{sect:dlog-inc-compute}, whilst leveraging that~$g'$ is small when implementing the arithmetic so as to reduce the quantum cost.
This results in an algorithm with similar performance and similar pre-computation as the one given in Sect.~\ref{sect:dlog-pre-compute}, but with the benefit of not requiring $g$ to generate~$\mathbb G$.

\subsection{Notes on computing multiple logarithms simultaneously}
As is the case for Shor's algorithm --- see e.g.~\cite{litinski} and~\cite{hyy23} --- we note that our extension of Regev's algorithm can be used to compute information on multiple discrete logarithms in each run by sampling vectors close to $\mathcal L^*_{x_1, \ldots, x_k, g} / \mathbb Z^d$, where
\begin{align*}
  \mathcal L_{x_1, \ldots, x_k, g}
  =
  \Bigg\{
    (z_1, \ldots, z_d) \in \mathbb Z^d
    \: \Bigg| \:
    g^{z_d}
    \cdot
    \prod_{i \, = \, 1}^{d-k} g_i^{z_i}
    \cdot
    \prod_{i \, = \, 1}^{k} x_i^{z_{d-k+i}}
    =
    1
  \Bigg\}
\end{align*}
and $x_i = g^{e_i}$ for all $i \in [1, k] \cap \mathbb Z$, instead of sampling vectors close to $\mathcal L^*_{x, g} / \mathbb Z^d$.

\section{Practical considerations}
In this section, we discuss various practical considerations when implementing both Regev's original algorithm and our extensions of his algorithm.

\subsection{Efficiency in implementations}
\label{sect:practical-efficiency}
As stated in the introduction, the quantum circuit for Regev's factoring algorithm, and for our extensions of it, is asymptotically smaller than the corresponding circuits for Shor's algorithms and the various variations thereof that are in the literature --- hereinafter referred to as the ``existing algorithms''.

Even so, for concrete problem instances of limited size, the quantum circuit for Regev's algorithm, and for our extensions of it, may in practice be significantly larger than optimized circuits for the existing algorithms.
This when using the space-saving arithmetic of Ragavan and Vaikuntanathan~\cite{rv23}.
It all comes down to constants.

Furthermore, it is not clear that it is the circuit size, i.e.~the gate count, that is the best metric whereby to compare quantum algorithms.
Other relevant metrics include but are not limited to the circuit depth, space usage, or volume.
It also matters what kind of large-scale fault-tolerant quantum computer one envisages, and what its architectural constraints are, and so forth.

The limiting factor for when a quantum algorithm can first conceivably be run on a future large-scale fault-tolerant quantum computer is arguably the cost of the circuit as it determines the cost per run.
But if and when such computers become more commonly available, the overall cost across all runs required will also be important to take into account.

\vspace{2mm}
\subsubsection{What problem instances are of key interest?}
Arguably, the key reason for why there is a large interest in quantum algorithms for factoring integers and computing discrete logarithms is that such algo\-rithms may be used to break currently widely deployed asymmetric cryptography.

Hence, the concrete efficiency of the aforementioned quantum algorithms with respect to cryptographically relevant problem instances is potentially much more interesting than their asymptotic efficiency.
This is especially true when one considers the fact that asymmetric cryptography based on the integer factoring and discrete logarithm problems is being phased out.
Hence, we are most likely not going to see much larger instances of these problems being used in the future, and so our focus should be on~$n$ in the range from say 2048~bits up to~4096 bits.

\vspace{2mm}

\subsubsection{Notes on the constant~$C$}
In order to even begin to compare Regev's algorithm and our extensions of it to the existing algorithms, the constant~$C$ needs to be fixed since it directly affects the circuit size and depth.

Our simulations~\cite{eg23-simulations} show that selecting $C \approx 2$ is sufficient for $n = 2048$~bit moduli when $d = \ceil{\sqrt{n}}$ and $m = d + 4$ runs are performed.
The success probability after $m$ runs is then close to one.
Ragavan and Vaikuntanathan~\cite[App.~A]{rv23} show, based on a heuristic assumption, that selecting $C = 3 + \epsilon + o(1)$ is sufficient asymptotically as $n \rightarrow \infty$.
This when using the LLL~\cite{lll} lattice basis reduction algorithm.
By instead using a better reduction algorithm such as BKZ~\cite{bkz}, by increasing~$d$ and/or by increasing the number of runs~$m$, it is sufficient to use a somewhat smaller~$C$, as corroborated by our simulations.

It is also worth noting that performing $m \ge d + 4$ runs is not strictly necessary for our extended algorithms to work.
We have only inherited this bound from Regev's analysis.
With sufficiently large~$C$, it should be sufficient to perform significantly fewer runs, but using a larger~$C$ increases the circuit size and depth.
Therefore, it is probably preferable to use a smaller $C$, and to run the quantum algorithm $m \ge d$ times.

\subsubsection{Notes on the choice of arithmetic}
Our extensions of Regev's algorithm can be implemented both with Regev's original arithmetic~\cite{regev23} and with the space-saving arithmetic of Ragavan and Vaikuntanathan~\cite{rv23}.
In Theorem~\ref{thm:main}, and in Theorem \ref{thm:order}-\ref{thm:phi-N} in App.~\ref{sect:order-finding-factoring}, we express the cost in terms of using the space-saving arithmetic since it is asymptotically on par with Regev's original arithmetic\footnote{When ignoring constant factors.} in terms of the circuit size, but uses less space.

It is however not clear that the space-saving arithmetic is better than Regev's original arithmetic in practice for cryptographically relevant problem instances:

Regev's original arithmetic yields a smaller circuit size, and a lower circuit depth, when accounting for constants, at the expense of using more space.
If space is cheap, it may be quite competitive.
Note also that most of the additional space required is not computational space, but rather space where quantum information can be stored until it is eventually needed for the uncomputation.

\subsubsection{Notes on optimizations}
There is a plethora of optimizations for the existing algorithms already in the literature, many of which can be combined.

Not all of these optimizations carry over to Regev's algorithm, and less time has been spent on seeking to optimize Regev's algorithm, further complicating the task of comparing the algorithms.
The same holds true for the extensions of Regev's algorithm introduced in this work.

\subsection{Handling error correction failures}
\label{sect:analysis-robustness}

Quantum computers as currently envisaged are inherently noisy, necessitating the use of some form of quantum error correction to achieve a sufficient level of fault tolerance to run complex quantum algorithms.

The quantum error correction is parameterized so as to achieve a certain lower bound on the probability of all errors that arise during the run being corrected\footnote{Under a set of assumptions, e.g.~on how the errors arise and their correlation.}, and hence of the output from the run being good.
The higher the bound, the more costly the error correction.

For quantum algorithms that need only yield a single good output --- such as Shor's algorithms~\cite{shor94, ekera-success, ekera-completely}, and Ekerå--Håstad's variations~\cite{ekera-hastad, ekera-pp, ekera-short-success} thereof when not making tradeoffs --- it may be advantageous on average to select a lower bound in the error correction at the expense of potentially having to re-run the algorithm if the output is bad.
Indeed, this approach was used in recent cost estimates~\cite{gidney-ekera} for the aforementioned algorithms.

For Regev's algorithm~\cite{regev23}, and the extensions thereof introduced in this work, the situation is different however in that the post-processing is only guaranteed to succeed with probability at least $1/4$ if all of the $m \ge d + 4$ runs yield good outputs, where we recall that $d = \ceil{\sqrt{n} \,}$ so~$d$ grows fairly rapidly in~$n$.

If the post-processing was to fail as soon as the output from a single bad run is included in the set of vectors fed to it, then --- short of us being able to efficiently distinguish the vectors output by good runs from those output by bad runs --- we would need to parameterize the error correction so that we have a sufficiently high probability of being able to efficiently construct a set of $m$ vectors yielded by good runs.
In turn, this would drive up the cost of the error correction.
Possible options for constructing such a set include but are not limited to exhausting subsets, or somehow efficiently filtering out good runs from bad runs via a distinguisher.

This being said, it turns out that Regev's original post-processing is in fact relatively robust to errors.
It continues to work even if a relatively large fraction of the vectors input to it are sampled from a different distribution --- provided that it is still fed a sufficiently large number of vectors that are yielded by good runs of the quantum algorithm, and that the constant~$C$ is sufficiently large.

We evidence this by means of simulations~\cite{eg23-simulations}, in which we test the behavior of Regev's classical post-processing when a fraction of the vectors are yielded by bad runs, as simulated by sampling these vectors from the uniform distribution.

In this section, we further corroborate these simulations by providing an analysis of the post-processing in the setting where a fraction of the vectors fed to it are yielded by bad runs of the quantum algorithm.

\subsubsection{Analysis of the robustness of the post-processing}
Suppose that the post-processing algorithm is fed $m$ vectors $\vect{w}_1, \ldots, \vect{w}_m$ generated in $m$ runs of the quantum algorithm.
However, only $d+4 \le m_1 \le m$ of these vectors are from good runs, and hence guaranteed to be at most a distance $\delta$ from $\mathcal L^*/\mathbb Z^d$.
The remaining $m_2 = m-m_1$ vectors are from bad runs, and hence sampled from some other unknown distribution~$\mathcal F$.
Suppose furthermore that we can not directly distinguish between vectors from good and bad runs, respectively.

The goal of the post-processing algorithm is to recover a basis for $\mathcal L$.
In what follows, we show, based on some assumptions on $\mathcal L$ and $\mathcal F$, that the post-processing succeeds in recovering a basis for~$\mathcal L$.

However, even without these assumptions, it can be seen that at least some vectors from $\mathcal L$ can be recovered.
We therefore begin by detailing why a single non-zero vector is recoverable, so as to provide partial motivation for why we deem the assumptions about~$\mathcal L$ and~$\mathcal F$ to be reasonable, and because it may be useful in the context of the algorithm in Sect.~\ref{sect:dlog-pre-compute}.

\paragraph{Recovering a single vector from $\mathcal L$}

As in the analysis of the original post-processing in Lemma~\ref{lem:4.4}, we consider a lattice~$\mathcal L''$ generated by the rows of
\[\arraycolsep=7pt\def\arraystretch{2}
\left(
\begin{array}{c|c c c}
  \vect{I}_{d\times d} & S \cdot \vect{w}_1^{\mathrm T} & \cdots & S \cdot \vect{w}_m^{\mathrm{T}} \\
  \hline
  \vect{0}_{m \times d} & & S \cdot \vect{I}_{m\times m} & \\
\end{array}
\right)
\]
for $S = \delta^{-1}$ a scaling parameter.

An unknown subset of $m_2$ of the coordinates in this lattice are dependent on the distribution $\mathcal F$.
Meanwhile, the other $d+m_1$ coordinates form a lattice $\mathcal L'$ that is of exactly the same form as the lattices considered in Lemma~\ref{lem:4.4}.

By Lemma~\ref{lem:4.4}, the first $d$ coordinates of any non-zero vector in~$\mathcal L'$ shorter than $S \varepsilon/2$ is a vector in $\mathcal L$, where $\varepsilon = (4\det \mathcal L)^{-1/m_1}/3$.
Furthermore, for any vector in~$\mathcal L''$, if the first $d$ coordinates are zero, the remaining coordinates are either~$0$, or a multiple of~$S$ which is larger than $S\varepsilon/2$.
Therefore, the first $d$ coordinates of any non-zero vector in $\mathcal L''$ that is shorter than ${S\varepsilon/2}$ is guaranteed to be a non-zero vector in $\mathcal L$.

The lattice $\mathcal L''$ has determinant $S^m$.
By Minkowski's first theorem, we are guaranteed that it contains a vector of length at most $\sqrt{m+d} \cdot S^{m/(m+d)}$.

Thus, if
\[
  \sqrt{m+d} \cdot S^{m/(m+d)} < S \cdot \varepsilon/2,
\]
which is true if $S$ is sufficiently large, we are guaranteed that the first $d$ coordinates of the shortest non-zero vector in $\mathcal L''$ is a vector in $\mathcal L$.

Furthermore, for somewhat larger~$S$, the first $d$ coordinates of vectors that are significantly larger than the shortest non-zero vectors in $\mathcal L''$ are also guaranteed to be in~$\mathcal L$.
It follows that we can recover vectors in~$\mathcal L$ simply by using an efficient lattice reduction algorithm such as LLL~\cite{lll}.

Hence, even without making additional assumptions on~$\mathcal L$ or~$\mathcal F$, we can guarantee that at least some non-zero vector from $\mathcal L$ can be recovered efficiently.

This guarantee is almost sufficient for us to be able to solve the discrete logarithm problem with pre-processing: As in Sect.~\ref{sect:dlog-pre-compute}, consider the lattice
\[
  \mathcal L
  =
  \mathcal L_x
  =
  \Bigg\{
    (z_1, \ldots, z_d) \in \mathbb Z^d
    \: \Bigg| \:
    x^{z_d} \prod_{i \, = \, 1}^{d-1} g_i^{z_i} = 1,
  \Bigg\},
\]
where we are given $g$, $x$ and $e_1, \ldots, e_{d-1}$ such that $g_i = g^{e_i}$, and are to compute~$e$ such that $x = g^e$.
For $\vect{z} = (z_1, \ldots, z_d)$ a non-zero vector recovered from this lattice~$\mathcal L_x$, we have that
\[
  0 = ez_d + \sum_{i \, = \, 1}^{d-1}e_iz_i \:\: (\text{mod } r)
\]
for $r$ the order of $g$.
It is then sufficient that~$z_d$ does not share any large factors with~$r$ for the discrete logarithm~$e$ to be recoverable.
To suppose that this is the case with noticeable probability for the recovered vector is indeed a very weak assumption, yet this assumption is sufficient to recover the discrete logarithm with pre-computation even if some runs are bad.

To show that the post-processing robustly succeeds in recovering a basis for~$\mathcal L$ we do, however, require a stronger assumption on $\mathcal L$ and $\mathcal F$.

\paragraph{Recovering a basis for~$\mathcal L$}
As previously mentioned, to show that we are able to recover a basis for~$\mathcal L$, we need to make some assumptions on~$\mathcal L$ and~$\mathcal F$.

Note that even if all~$m$ runs are good, we cannot guarantee that the post-processing succeeds in recovering a basis for~$\mathcal L$ without making some additional assumption on~$\mathcal L$.
In particular, Theorem~\ref{thm:main} requires a specific lattice to have a short basis, and therefore needs to use Assumption~\ref{assump:basis-generic}.
For the analysis in this section, we make use of a similar assumption so as to ensure that the short vectors in~$\mathcal L''$ yield vectors that generate~$\mathcal L$, which in turn ensures that these vectors can be used to recover a basis for~$\mathcal L$.

To this end, we first note that if the first~$d$ coordinates of a vector $\vect{u}'' \in \mathcal L''$ form a vector $\vect{u} \in \mathcal L$, the~$m_1$ coordinates of $\vect{u}''$ that correspond to good runs have absolute value at most $\norm{\vect{u}}S\delta = \norm{\vect{u}}$.
The remaining~$m_2$ coordinates are congruent to $S\vect{u} \cdot \vect{w}_i$ modulo $S$, where $\vect{w}_i$ is sampled from~$\mathcal F$.
We do not have much control over the values of these $m_2$ coordinates, but we can expect them to be small for at least some short vectors $\vect{u}$.

To analyze the vectors in $\mathcal L''$ that are such that all coordinates are small, and such that the first~$d$ coordinates form a vector in~$\mathcal L$, we consider the sublattice of~$\mathcal L''$ where the first~$d$ coordinates are exactly the vectors in $\mathcal L$.
This sublattice is thus generated by the rows of
\[\arraycolsep=7pt\def\arraystretch{2}
  \left(
  \begin{array}{c|c c c}
    \vect{B} & S \vect{B}\cdot \vect{w}_1^{\mathrm T} & \cdots & S\vect{B} \cdot \vect{w}_m^{\mathrm{T}} \\
    \hline
    \vect{0}_{m \times d} & & S \cdot \vect{I}_{m\times m} & \\
  \end{array}
  \right)
\]
where $\vect{B}$ is a basis for $\mathcal L$.
For a small vector $\vect{u} \in \mathcal L$, we have already seen that the size of the coordinates that correspond to good runs is limited.
We therefore consider a related lattice $\Lambda$ that do not contain the $m_1$ coordinates that correspond to good runs.
With $\vect{F}$ a $d\times m_2$-dimensional matrix such that the columns are equal to the $m_2$ vectors $\vec w_i$ that correspond to bad runs, we can see that this lattice $\Lambda$ is generated by the rows of
\begin{align}\label{eq:lambda-form}
  \arraycolsep=7pt\def\arraystretch{2}
  \left(
  \begin{array}{c|c}
    \vect{B} & S\vect{B} \cdot \vect{F} \\
    \hline
    \vect{0}_{m_2\times d} & S \cdot \vect{I}_{m_2} \\
  \end{array}
  \right).
\end{align}

The first~$d$ coordinates of a short vector $\vect{x} \in \Lambda$ is a short vector $\vect{u} \in \mathcal L$.
As such, $\vect{x}$ directly corresponds to a short vector in $\mathcal L''$, where the remaining $m_1$ coordinates not included in $\Lambda$ have absolute value at most $\norm{\vect{u}} \le \norm{\vect{x}}$.
Thus, the full vector in~$\mathcal L''$ is not longer than ${\norm{\vect{x}}(1 + m_1)}$.
As such, the vectors in~$\Lambda$ of length at most~$T$ correspond to vectors in~$\mathcal L''$ of length at most $T(1+ m_1)$.

By using Claim~\ref{cla:5.1}, we recover a set $\mathcal S$ of vectors in $\mathcal L''$ that generate all vectors in $\mathcal L''$ of length at most $T(1+ m_1)$.
The relevant coordinates of these recovered vectors thus generate all vectors in $\Lambda$ of length at most $T$.

Next, we show that the relevant coordinates of the vectors in~$\mathcal S$ only generate vectors in $\Lambda$.
Each vector in $\mathcal S$ has length at most $2^{(d+m)/2} \cdot \sqrt{d+m} \cdot T(1+m_1)$.
Meanwhile, as previously noted, we are guaranteed that any vector in~$\mathcal L''$ shorter than~$S \varepsilon/2$ corresponds to a vector in $\mathcal L$, and therefore to a vector in $\Lambda$.

Thus, if
\begin{align} \label{eq:required-inequality}
    2^{(d+m)/2} \cdot \sqrt{d+m} \cdot T(1+m_1) < S \varepsilon/2
\end{align}
the relevant coordinates of the vectors in $\mathcal S$ generate a sublattice of $\Lambda$, and we can thus recover a basis for this sublattice.

In summary, if~(\ref{eq:required-inequality}) holds, we can recover vectors in $\mathcal L''$ that are such that the relevant coordinates generate a sublattice of $\Lambda$.
Any vector $\vect{x} \in \Lambda$ such that ${\norm{\vect{x}} \le T}$ is guaranteed to be in this sublattice.
Furthermore, as the first $d$ coordinates of any vector in $\Lambda$ is a vector in $\mathcal L$, we can easily recover a related sublattice of $\mathcal L$ from the recovered vectors in $\mathcal L''$.
Thus, in order for this process to recover a basis for $\mathcal L$, it is sufficient that the vectors in $\Lambda$ of length at most $T$ are guaranteed to contain sufficient information about $\mathcal L$.

\paragraph{Required properties of $\Lambda$}

Proving that the vectors in $\Lambda$ shorter than $T$ are sufficient to recover a basis for $\mathcal L$ seems hard, and would presumably require significant assumptions to be made regarding both $\mathcal L$ and $\mathcal F$.

We therefore instead choose to simply assume that this is the case for lattices~$\Lambda$ of this form with some noticeable probability~$p$ over the randomness in~$\mathcal F$.
This is equivalent to assuming that~$\Lambda$ has a basis where all information about~$\mathcal L$ is contained in short basis vectors.

To determine for which bound $T$ it is reasonable to assume that such a basis exists, we note that $\Lambda$ has determinant $S^{m_2}\det(\mathcal L)$.
Therefore, by Minkowski's first theorem we are guaranteed that $\Lambda$ contains a vector of length at most
\[
  {\sqrt{m_2 + d} \cdot S^{m_2/(m_2+d)}\det(\mathcal L)^{1/(m_2+d)}}.
\]

A random lattice is expected to have a basis where each basis vector is only slightly longer than the shortest vector in the lattice.
This leads us to make the following concrete assumption:
\begin{assumption}\label{assump:error-basis}
    Let $\vect{B}$ be a basis for $\mathcal L$, $F > 0$ be a constant and $\vect{F}$ be a $d\times m_2$-dimensional matrix where each of the $m_2$ columns are sampled from $\mathcal F$.
    Then, with probability $p$ the lattice $\Lambda$ generated by the rows of
    \begin{align*}
        \arraycolsep=7pt\def\arraystretch{2}
    \left(
    \begin{array}{c|c}
        \vect{B} & S\vect{B} \cdot \vect{F} \\
        \hline
        \vect{0}_{m_2\times d} & S \cdot \vect{I}_{m_2} \\
    \end{array}
    \right)
    \end{align*}
    has a basis such that each basis vector with non-zero values in its first $d$ coordinates is shorter than $T = F\sqrt{m_2 + d} \cdot S^{m_2/(m_2+d)} \det(\mathcal L)^{1/(m_2+d)}$.
\end{assumption}

To further motivate this assumption, we note that if~$\mathcal L$ is a nice lattice, we would expect it to contain many short vectors, and we would expect these short vectors to span the full lattice~$\mathcal L$.
We furthermore expect many of these short vectors in~$\mathcal L$ to have an inner product far from an integer with the columns of~$\vect{F}$ that are sampled from~$\mathcal F$.
These short vectors in~$\mathcal L$ therefore do not correspond to short vectors in~$\Lambda$.
A small fraction of the vectors may however have an inner product close to an integer with each of the columns of~$\vect F$.

A potential risk is hence that the distribution~$\mathcal F$ may somehow be biased so as to cause some specific part of~$\mathcal L$ to be more likely to have an inner product far from an integer with vectors sampled from~$\mathcal F$, and therefore to be more likely to correspond to large vectors in~$\Lambda$ leading to only a part of the lattice~$\mathcal L$ being recovered.
If no such bias exists, it is, however, natural to assume that the short vectors in~$\mathcal L$ that have an inner product close to an integer with each of the columns of~$\vect F$ still span the full lattice~$\mathcal L$.

It is hence plausible that Assumption~\ref{assump:error-basis} holds with probability close to~$1$ if the lattice $\mathcal L$ and the distribution $\mathcal F$ are nice.
In particular, if there exists a short basis for~$\mathcal L$, and if~$\mathcal F$ is the uniform distribution, we expect Assumption~\ref{assump:error-basis} to hold with probability~$p$ close to~$1$, and this is also corroborated by our simulations.

Theorem~\ref{thm:post-process-robust} below follows from this assumption and the above analysis.

\begin{theorem}[Robust post-processing derived from Theorem~1.1~in~\cite{regev23}]
  \label{thm:post-process-robust}
  Let $\mathcal L$ be a $d$-dimensional lattice with $\det \mathcal L < 2^{d^2}$, and let $m, m_1$ and $m_2$, and the vectors $\vect{w}_1, \ldots, \vect{w}_m$, be as above.
  Furthermore, let $\delta = 2^{-Cd}\cdot \sqrt{d}/2$ and $S = \delta^{-1}$ for constant
  \begin{align*}
      C > \left(\dfrac{5}{2} + \dfrac{m}{2d}\right)\left(1  + \dfrac{m_2}{d}\right) + o(1).
  \end{align*}

  Then, under Assumption~\ref{assump:error-basis}, there is an efficient classical algorithm that, with probability at least $p / 4$ over the choice of the~$\vect{w}_i$, recovers a basis for~$\mathcal L$.
\end{theorem}

\begin{proof}
  The theorem follows from the above analysis, where the success probability bound of $p/4$ is obtained by combining the probability~$p$ from Assumption~\ref{assump:error-basis} with the bound of~$1/4$ from Lemma~\ref{lem:4.4}.
  Furthermore, the bound on~$C$ is obtained by inserting the values of~$\varepsilon$ and~$T$ in~(\ref{eq:required-inequality}) and taking logarithms, leading to the requirement that
  \begin{align*}
    (d+m)/2 + \dfrac{Cdm_2 + \log(\det\mathcal L)}{m_2+d} + o(d) < Cd - \dfrac{\log(\det\mathcal L)}{m_1}.
  \end{align*}

  By using that $\det(\mathcal L) < 2^{d^2}$ and $m_1 \ge d$, this may be simplified to
  \begin{align*}
    \dfrac{5}{2} + \dfrac{m}{2d} + \dfrac{Cm_2}{m_2+d} + o(1) < C,
  \end{align*}
  leading to the requirement
  \begin{align*}
    C > \left(\dfrac{5}{2} + \dfrac{m}{2d}\right)\left(1  + \dfrac{m_2}{d}\right) + o(1)
  \end{align*}
  and so the theorem follows.
\end{proof}

\section*{Acknowledgments}
\anonymize{We are grateful to Johan Håstad, Oded Regev, and the participants of the Quantum Cryptanalysis seminar at Schloss Dagstuhl, for useful comments.
Martin Ekerå thanks Schloss Dagstuhl and the organizers of the seminar for creating an environment where scientific progress is facilitated.
Funding and support for this work was provided by the Swedish NCSA that is a part of the Swedish Armed Forces.}

\appendix
\section{Order finding and factoring}
\label{sect:order-finding-factoring}
Let~$\mathbb G$ be a finite Abelian group and let $g \in \mathbb G$.
Suppose that our goal is to find the order~$r$ of~$g$.
Let $g_1, \ldots, g_{d-1}$ be $d-1$ small elements in~$\mathbb G$.

In particular, for $\mathbb G = \mathbb Z_N^*$, for~$N$ a positive $n$-bit integer that is coprime to the first~$d = \ceil{\sqrt{n} \,}$ primes, we take $g_1, \ldots, g_{d-1}$ to be the first $d-1$ primes that when perceived as elements of~$\mathbb Z_N^*$ are distinct from~$g$.

By Lemma~\ref{lem:quantum-algorithm-Zn}, there exists an efficient quantum algorithm that outputs a vector $\vect{w} \in [0, 1)^d$ that, except for with probability $1 / \poly(d)$, is within distance $\sqrt{d/2} \cdot 2^{-C\sqrt{n}}$ of a uniformly chosen $\vect{v} \in \mathcal L_{g}^*/\mathbb Z^d$, where
\begin{align*}
  \mathcal L_{g}
  =
  \Bigg\{
  (z_1, \ldots, z_d) \in \mathbb Z^{d}
    \: \Bigg| \:
    g^{z_{d}} \prod_{i \, = \, 1}^{d-1} g_i^{z_i} = 1
  \Bigg\},
\end{align*}
for $C > 0$ some constant.
By performing $m \ge d + 4$ runs of this quantum algorithm, we obtain~$m$ such vectors $\vect{w}_1, \ldots, \vect{w}_m$.

By Lemma~\ref{lem:post-process}, there is an efficient classical algorithm that, with probability at least $1/4$, recovers a basis for a sublattice~$\Lambda$ of~$\mathcal L_g$ given the vectors $\vect{w}_1, \ldots, \vect{w}_m$.
This sublattice~$\Lambda$ contains all vectors in~$\mathcal L_g$ of norm at most~$T = \exp(\ordo(\sqrt n))$.

Under Assumption~\ref{assump:basis-generic}, there exist a basis of~$\mathcal L_g$ with all basis vector having norm at most $T$.
It then follows that~$\Lambda = \mathcal L_{g}$, and Lemma~\ref{lem:post-process} thus recovers a basis~$\vect{B}$ for $\mathcal L_{g}$.
Furthermore, $\vect {v}_r = (0, \ldots, 0, r)$ is in~$\mathcal L_{g}$, and by definition~$r$ is the least positive integer such that $g^r = 1$.
Therefore, $\vect {v}_r$ is the shortest vector in~$\mathcal L_{g}$ that is non-zero only in the last coordinate.
Since we know a basis~$\vect B$ for~$\mathcal L_{g}$, we can efficiently recover~$\vect{v}_r$ and by extension~$r$.

We have heuristically verified that the above procedure yields a basis of~$\mathcal L_{g}$, and that we can solve for the order~$r$, by means of simulations.
For $n = 2048$~bits, the simulations indicate that it suffices to take $C \approx 2$, and that the success probability is close to one after $d + 4$ runs of the quantum algorithm.

We are now ready to summarize the above analysis in a theorem:
\begin{theorem}
  \label{thm:order}
  Let~$N$ be a positive $n$-bit integer and let $d = \ceil{\sqrt{n} \,}$.
  Let $g \in \mathbb Z_N^*$.
  As in~\cite[Theorem~1]{rv23}, let~$G$ be the gate cost of a quantum circuit that takes
  \begin{align*}
    \ket{a, b, t, 0^S} \rightarrow \ket{a, b, (t + ab) \text{ mod } N, 0^S}
  \end{align*}
  for $a, b, t \in [0, N) \cap \mathbb Z$ and~$S$ the number of ancilla qubits required.
  Then, under Assumption~\ref{assump:basis-generic}, there is an efficient classical algorithm that, by calling a quantum circuit $d + 4$ times, yields the order~$r$ of~$g$ with probability at least~$1/4$.
  This quantum circuit has gate cost $\ordo(n^{1/2} \, G + n^{3/2})$, and it requires
  \begin{align*}
    S + \left( \frac{C}{\log \phi} + 8 + o(1) \right) n
  \end{align*}
  qubits of space, for some constant $C > 0$ and~$\phi$ the golden ratio.
\end{theorem}
\begin{proof}
  The proof follows from the above analysis, and from Lemma~\ref{lem:quantum-algorithm-Zn} where the gate and space costs of the quantum circuit are analyzed.

  Note that if~$N$ is divisible by one or more of the first~$d$ primes, then these prime powers may be factored out before quantum order finding is performed.

  More specifically, order finding may be performed efficiently classically with respect to these prime powers, after which the partial results may be efficiently combined classically to yield the order~$r$ of~$g$.
\end{proof}

\subsection{Factoring~$N$ via order finding in~$\mathbb Z_N^*$}
\label{sect:factor-via-r}
Let~$N$ be a positive composite integer.
Suppose that we pick~$g$ uniformly at random from $\mathbb Z_N^*$ and compute the order~$r$ of~$g$ by using the algorithm in App.~\ref{sect:order-finding-factoring}.

Then, with very high probability --- that is lower-bounded and shown to tend to one asymptotically in~\cite{ekera-completely, ekera-success} --- we can completely factor~$N$ given~$r$ via the procedure in~\cite{ekera-completely}.
This provides an alternative to Regev's factoring algorithm, that yields the complete factorization of~$N$ in $d + 4$ runs by factoring via order finding, at the expense of making a stronger heuristic assumption, and at the expense of including one element that is not small in the product that is computed quantumly.

\subsection{Finding the order of an Abelian group}
Let~$\mathbb G$ be a finite Abelian group.
Suppose that our goal is to find the order~$\# \mathbb G$ of~$\mathbb G$.
Let $g_1, \ldots, g_{d}$ be $d$ small elements in~$\mathbb G$.
In particular, for $\mathbb G = \mathbb Z_N^*$, for~$N$ a positive $n$-bit integer that is coprime to the first $d = \ceil{\sqrt{n} \,}$ primes, we take $g_1, \ldots, g_{d} \in \mathbb Z_N^*$ to be the first~$d$ primes perceived as elements of~$\mathbb Z_N^*$.

By Lemma~\ref{lem:quantum-algorithm-Zn}, there exists an efficient quantum algorithm that outputs a vector $\vect{w} \in [0, 1)^d$ that, except for with probability $1 / \poly(d)$, is within distance ${\sqrt{d/2} \cdot 2^{-C\sqrt{n}}}$ of a uniformly chosen $\vect{v} \in \mathcal L^*/\mathbb Z^d$, where
\begin{align*}
  \mathcal L
  =
  \Bigg\{
    (z_1, \ldots, z_d) \in \mathbb Z^{d}
    \: \Bigg| \:
    \prod_{i \, = \, 1}^{d} g_i^{z_i} = 1
  \Bigg\},
\end{align*}
for $C > 0$ some constant.
By performing $m \ge d + 4$ runs of this quantum algorithm, we obtain~$m$ such vectors $\vect{w}_1, \ldots, \vect{w}_m$.
By Lemma~\ref{lem:post-process}, there is an efficient classical algorithm that, with probability at least $1/4$, recovers a basis~$\vect{B}$ of a sublattice~$\Lambda$ of~$\mathcal L_g$ when given the vectors $\vect{w}_1, \ldots, \vect{w}_m$.
Furthermore, the sublattice~$\Lambda$ of~$\mathcal L$ contains all vectors in~$\mathcal L$ of norm at most~$T = \exp(\ordo(\sqrt{n}))$.

Under Assumption~\ref{assump:basis-generic} there exist a basis for~$\mathcal L$ with all basis vectors having norm at most $T$.
It then follows that $\Lambda = \mathcal L$, and~$\vect{B}$ is thus a basis for~$\mathcal L$.
Given~$\vect B$, we can efficiently compute $\det \mathcal L = | \det \vect{B} \,|$.
By Lemma~\ref{lem:determinant} and Assumption~\ref{assump:basis-generic}, we then have that $\# \mathbb Z_N^* = \varphi(N) = \det \mathcal L$, where~$\varphi$ is Euler's totient function.

We have heuristically verified that the above procedure yields a basis of~$\mathcal L$, and that we can solve for $\varphi(N)$, by means of simulations.
For $n = 2048$ bits, the simulations indicate that it suffices to take $C \approx 2$, and that the success probability is close to one after $d + 4$ runs of the quantum algorithm.

We are now ready to summarize the above analysis in a theorem:
\begin{theorem}
  \label{thm:phi-N}
  Let~$N$ be a positive $n$-bit integer coprime to the first $d = \ceil{\sqrt{n} \,}$ primes, and let $g_1, \ldots, g_{d} \in \mathbb Z_N^*$ be the first~$d$ primes perceived as elements of~$\mathbb Z_N^*$.
  As in~\cite[Theorem~1]{rv23}, let~$G$ be the gate cost of a quantum circuit that takes
  \begin{align*}
    \ket{a, b, t, 0^S} \rightarrow \ket{a, b, (t + ab) \text{ mod } N, 0^S}
  \end{align*}
  for $a, b, t \in [0, N) \cap \mathbb Z$ and~$S$ the number of ancilla qubits required.
  Then, under Assumption~\ref{assump:basis-generic}, there is an efficient classical algorithm that, by calling a quantum circuit $d + 4$ times, yields $\varphi(N)$ with probability at least~$1/4$.
  This quantum circuit has gate cost $\ordo(n^{1/2} \, G + n^{3/2})$, and it requires
  \begin{align*}
    S + \left( \frac{C}{\log \phi} + 8 + o(1) \right) n
  \end{align*}
  qubits of space, for some constant $C > 0$ and~$\phi$ the golden ratio.
\end{theorem}
\begin{proof}
  The proof follows from the above analysis, and from Lemma~\ref{lem:quantum-algorithm-Zn} where the gate and space costs of the quantum circuit are analyzed.
\end{proof}

Note that if~$N$ is divisible by one or more of the first~$d$ primes in Theorem~\ref{thm:phi-N} above, then these prime powers may be factored out before calling the quantum algorithm.
The value of~$\varphi$ may then be efficiently computed classically with respect to these prime powers, and all partial results efficiently combined classically to yield~$\varphi(N)$.
The restriction imposed in Theorem~\ref{thm:phi-N} that~$N$ must be coprime to the first~$d$ primes does hence not imply a loss of generality.

\subsection{Factoring~$N$ by finding~$\varphi(N) = \# \mathbb Z_N^*$}
\label{sect:factor-via-phi}
Given~$\varphi(N)$, we may use a randomized version of Miller's algorithm~\cite{miller76} to factor~$N$ completely as explained in~\cite{ekera-completely}.
This provides yet another alternative for factoring via order finding that yields the complete factorization of~$N$ in $d + 4$ runs at the expense of making a stronger heuristic assumption.

A clear advantage of this alternative compared to that in App.~\ref{sect:factor-via-r} is that all elements that are exponentiated quantumly are small, bringing the quantum cost essentially on par with that of Regev's factoring algorithm~\cite{regev23}.

In fact, the quantum cost is slightly lower than than of Regev's algorithm:
Whereas Regev's algorithm exponentiates the squares of the first~$d$ primes, the above algorithm exponentiates the first~$d$ primes, so the numbers that are exponentiated are slightly smaller.

\subsection{Notes on generalizations}
The above theorems are for algorithms specific to~$\mathbb Z_N^*$ since there is a notion of small group elements in~$\mathbb Z_N^*$, and since~$\mathbb Z_N^*$ is extensively used in cryptography.
As previously stated, the algorithms may be generalized to other Abelian groups, but for the algorithms to have an advantage over other algorithms in the literature there must exist a notion of small elements in the group.

\end{document}